\documentclass{elsarticle}
\usepackage[T1]{fontenc}
\usepackage{graphicx}
\usepackage{epsfig}
\usepackage{amssymb,amsmath,amsthm}
\usepackage{url}
\usepackage[colorlinks=true,citecolor=black,linkcolor=black,urlcolor=blue]{hyperref}

\def\paren#1{\left( #1 \right)}
\def\acc#1{\left\{ #1 \right\}}

\renewcommand{\le}{\leqslant}
\renewcommand{\ge}{\geqslant}

\newcommand{\arxiv}[1]{\href{http://arxiv.org/abs/#1}{\texttt{arXiv:#1}}}
\newtheorem{theorem}{Theorem}

\newtheorem{corollary}[theorem]{Corollary}
\newtheorem{lemma}[theorem]{Lemma}

\newtheorem{conjecture}[theorem]{Conjecture}

\begin{document}
\begin{frontmatter}
\title{Avoidability of circular formulas}
\author[LIRMM]{Guilhem Gamard}
\author[LIRMM,CNRS]{Pascal Ochem}
\author[LIRMM,um3]{Gwena\"el Richomme}
\author[LIRMM,um3]{Patrice S\'e\'ebold}
\address[LIRMM]{LIRMM, Universit\'e de Montpellier and CNRS, France}
\address[CNRS]{CNRS}
\address[um3]{Universit\'e Paul-Val\'ery Montpellier 3}

\begin{abstract}
Clark has defined the notion of $n$-avoidance basis which contains the avoidable formulas
with at most $n$ variables that are closest to be unavoidable in some sense.
The family $C_i$ of circular formulas is such that $C_1=AA$, $C_2=ABA.BAB$, $C_3=ABCA.BCAB.CABC$ and so on.
For every $i\le n$, the $n$-avoidance basis contains $C_i$.
Clark showed that the avoidability index of every circular formula and of every formula in
the $3$-avoidance basis (and thus of every avoidable formula containing at most 3 variables) is at most 4.
We determine exactly the avoidability index of these formulas.
\end{abstract}

\end{frontmatter}

\section{Introduction}\label{sec:intro}
A \emph{pattern} $p$ is a non-empty finite word over an alphabet
$\Delta=\acc{A,B,C,\dots}$ of capital letters called \emph{variables}.
An \emph{occurrence} of $p$ in a word $w$ is a non-erasing morphism $h:\Delta^*\to\Sigma^*$
such that $h(p)$ is a factor of $w$.
The \emph{avoidability index} $\lambda(p)$ of a pattern $p$ is the size of the
smallest alphabet $\Sigma$ such that there exists an infinite word
over $\Sigma$ containing no occurrence of $p$.
Bean, Ehrenfeucht, and McNulty~\cite{BEM79} and Zimin~\cite{Zimin}
characterized unavoidable patterns, i.e., such that $\lambda(p)=\infty$.
We say that a pattern $p$ is \emph{$t$-avoidable} if $\lambda(p)\le t$.
For more informations on pattern avoidability, we refer to Chapter 3 of Lothaire's book~\cite{Lothaire2002}.
See also this book for basic notions in Combinatorics on Words.

A variable that appears only once in a pattern is said to be \emph{isolated}.
Following Cassaigne~\cite{Cassaigne1994}, we associate to a pattern $p$ the \emph{formula} $f$
obtained by replacing every isolated variable in $p$ by a dot.
The factors between the dots are called \emph{fragments}.

An \emph{occurrence} of a formula $f$ in a word $w$ is a non-erasing morphism $h:\Delta^*\to\Sigma^*$
such that the $h$-image of every fragment of $f$ is a factor of $w$.
As for patterns, the avoidability index $\lambda(f)$ of a formula $f$ is the size of the
smallest alphabet allowing the existence of an infinite word containing no occurrence of $f$.
Clearly, if a formula $f$ is associated to a pattern $p$,
every word avoiding $f$ also avoids $p$, so $\lambda(p)\le\lambda(f)$.
Recall that an infinite word is \emph{recurrent} if every finite factor appears
infinitely many times.
If there exists an infinite word over $\Sigma$ avoiding $p$,
then there exists an infinite recurrent word over $\Sigma$ avoiding $p$.
This recurrent word also avoids $f$, so that $\lambda(p)=\lambda(f)$.
Without loss of generality, a formula is such that no variable is isolated
and no fragment is a factor of another fragment.

Cassaigne~\cite{Cassaigne1994} began and Ochem~\cite{Ochem2004}
finished the determination of the avoidability index of every pattern with at most 3 variables.
A \emph{doubled} pattern contains every variable at least twice.
Thus, a doubled pattern is a formula with exactly one fragment.
Every doubled pattern is 3-avoidable~\cite{O16}.
A formula is said to be \emph{binary} if it has at most 2 variables.
The avoidability index of every binary formula has been recently determined~\cite{OchemRosenfeld2016}.
We say that a formula $f$ is \emph{divisible} by a formula $f'$ if $f$ does not avoid $f'$,
that is, there is a non-erasing morphism $h$ such that the image of every fragment of $f'$ by $h$ is a factor of a fragment of $f$.
If $f$ is divisible by $f'$, then every word avoiding $f'$ also avoids $f$ and thus $\lambda(f)\le\lambda(f')$.
Moreover, the reverse $f^R$ of a formula $f$ satisfies $\lambda(f^R)=\lambda(f)$.
For example, the fact that $ABA.AABB$ is 2-avoidable implies that $ABAABB$ and $BAB.AABB$ are 2-avoidable.
See Cassaigne~\cite{Cassaigne1994} and Clark~\cite{Clark} for more information on formulas and divisibility.

Clark~\cite{Clark} has introduced the notion of \emph{$n$-avoidance basis} for formulas,
which is the smallest set of formulas with the following property:
for every $i\le n$, every avoidable formula with $i$ variables is divisible by at least one formula
with at most $i$ variables in the $n$-avoidance basis.

From the definition, it is not hard to obtain that the $1$-avoidance basis is $\acc{AA}$ and the $2$-avoidance basis is $\acc{AA, ABA.BAB}$.
Clark obtained that the $3$-avoidance basis is composed of the following formulas:
\begin{itemize}
 \item $AA$
 \item $ABA.BAB$
 \item $ABCA.BCAB.CABC$
 \item $ABCBA.CBABC$
 \item $ABCA.CABC.BCB$
 \item $ABCA.BCAB.CBC$
 \item $AB.AC.BA.CA.CB$
\end{itemize}

The following properties of the avoidance basis are derived.
\begin{itemize}
 \item The $n$-avoidance basis is a subset of the $(n+1)$-avoidance basis.
 \item The $n$-avoidance basis is closed under reverse. (In particular, $ABCA.BCAB.CBC$ is the reverse of $ABCA.CABC.BCB$.)
 \item Two formulas in the $n$-avoidance basis with the same number of variables are incomparable by divisibility. (However, $AA$ divides $AB.AC.BA.CA.CB$.)
 \item The $n$-avoidance basis is computable.
\end{itemize}

The \emph{circular formula} $C_t$ is the formula over $t\ge1$ variables $A_0,\ldots,A_{t-1}$ containing
the $t$ fragments of the form $A_iA_{i+1}\ldots A_{i+t}$ such that the indices are taken modulo $t$.
Thus, the first three formulas in the $3$-avoidance basis, namely $C_1=AA$, $C_2=ABA.BAB$, and $C_3=ABCA.BCAB.CABC$,
are also the first three circular formulas.
More generally, for every $t\le n$, the $n$-avoidance basis contains $C_t$.

It is known that $\lambda(AA)=3$~\cite{Thue06}, $\lambda(ABA.BAB)=3$~\cite{Cassaigne1994}, and $\lambda(AB.AC.BA.CA.CB)=4$~\cite{BNT89}.
Actually, $AB.AC.BA.CA.CB$ is avoided by the fixed point $b_4=0121032101230321\dots$ of the morphism given below.
$$\begin{array}{l}
\texttt{0}\mapsto\texttt{01}\\
\texttt{1}\mapsto\texttt{21}\\
\texttt{2}\mapsto\texttt{03}\\
\texttt{3}\mapsto\texttt{23}\\
\end{array}$$
Clark~\cite{Clark} obtained that $b_4$ also avoids $C_i$ for every $i\ge1$, so that $\lambda(C_i)\le4$ for every $i\ge1$.
He also showed that the avoidability index of the other formulas in the $3$-avoidance basis is at most~$4$.
Our main results finish the determination of the avoidability index of the circular formulas (Theorem~\ref{thm:circ}) and the formulas in the $3$-avoidance basis (Theorem~\ref{thm:two}).

\section{Conjugacy classes and circular formulas}
In this section, we determine the avoidability index of circular formulas.
\begin{theorem}~\label{thm:circ}
$\lambda(C_3)=3$. $\forall i\ge4$, $\lambda(C_i)=2$.
\end{theorem}
We consider a notion that appears to be useful in the study of circular formulas.
A \emph{conjugacy class} is the set of all the conjugates of a given word, including the word itself.
The length of a conjugacy class is the common length of the words in the conjugacy class.
A word contains a conjugacy class if it contains every word in the conjugacy class as a factor.
Consider the uniform morphisms given below.

\noindent
\begin{minipage}[b]{0.46\linewidth}
\centering
$$\begin{array}{l}
g_2(\texttt{0})=\texttt{0000101001110110100}\\
g_2(\texttt{1})=\texttt{0011100010100111101}\\
g_2(\texttt{2})=\texttt{0000111100010110100}\\
g_2(\texttt{3})=\texttt{0011110110100111101}\\
\end{array}$$
\end{minipage}
\begin{minipage}[b]{0.24\linewidth}
\centering
$$\begin{array}{l}
g_3(\texttt{0})=\texttt{0010}\\
g_3(\texttt{1})=\texttt{1122}\\
g_3(\texttt{2})=\texttt{0200}\\
g_3(\texttt{3})=\texttt{1212}\\
\end{array}$$
\end{minipage}
\begin{minipage}[b]{0.28\linewidth}
\centering
$$\begin{array}{l}
g_6(\texttt{0})=\texttt{01230}\\
g_6(\texttt{1})=\texttt{24134}\\
g_6(\texttt{2})=\texttt{52340}\\
g_6(\texttt{3})=\texttt{24513}\\
\end{array}$$
\end{minipage}
\\

\begin{lemma}\label{L:1}{\ }
\begin{itemize}
 \item The word $g_2(b_4)$ avoids every conjugacy class of length at least $5$.
 \item The word $g_3(b_4)$ avoids every conjugacy class of length at least $3$.
 \item The word $g_6(b_4)$ avoids every conjugacy class of length at least $2$.
\end{itemize}
\end{lemma}

\begin{proof}
We only detail the proof for $g_2(b_4)$, since the proofs for $g_3(b_4)$ and $g_6(b_4)$ are similar.
Notice that $g_2$ is $19$-uniform.
First, a computer check shows that $g_2(b_4)$ contains no conjugacy class of length $i$ with $5\le i\le55$ (i.e., $2\times19+17$). 

Suppose for contradiction that $g_2(b_4)$ contains a conjugacy class of length at least $56$ (i.e., $2\times19+18$).
Then every element of the conjugacy class contains a factor $g_2(ab)$ with $a,b\in\Sigma_4$.
In particular, one of the elements of the conjugacy class can be written as $g_2(ab)s$.
The word $g_2(b)sg_2(a)$ is also a factor of $g_2(b_4)$.
A computer check shows that for every letters $\alpha$, $\beta$, and $\gamma$ in $\Sigma_4$ such that
$g_2(\alpha)$ is a factor of $g_2(\beta\gamma)$, $g_2(\alpha)$ is either a prefix or a suffix of $g_2(\beta\gamma)$.
This implies that $s$ belongs to $g_2(\Sigma_4^+)$.

Thus, the conjugacy class contains a word $w=g_2(\ell_1\ell_2\ldots\ell_k)=x_1x_2...x_{19k}$.
Consider the conjugate $\tilde{w}=x_7x_8\ldots x_{19k}x_1x_2x_3x_4x_5x_6$.
Observe that the prefixes of length $6$ of $g_2(0)$, $g_2(1)$, $g_2(2)$, and $g_2(3)$ are different.
Also, the suffixes of length $12$ of $g_2(0)$, $g_2(1)$, $g_2(2)$, and $g_2(3)$ are different.
Then the prefix $x_7\ldots x_{19}$ and the suffix $x_1\ldots x_6$ of $\tilde{w}$
both force the letter $\ell_1$ in the pre-image. That is, $b_4$ contains $\ell_1\ell_2\ldots\ell_k\ell_1$.
Similarly, the conjugate of $w$ that starts with the letter $x_{19(r-1)+7}$ implies that $b_4$ contains $\ell_r\ldots\ell_k\ell_1\ldots\ell_r$.
Thus, $b_4$ contains an occurrence of the formula $C_k$.
This is a contradiction since Clark~\cite{Clark} has shown that $b_4$ avoids every circular formula $C_i$ with $i\ge1$.
\end{proof}

Notice that if a word contains an occurrence of $C_i$, then it contains a conjugacy class of length at least $i$. 
Thus, a word avoiding every conjugacy class of length at least $i$ also avoids every circular formula $C_t$ with $t\ge i$.
Moreover, $g_2(b_4)$ contains no occurrence of $C_4$ such that the length of the image of every variable is $1$.
By Lemma~\ref{L:1}, this gives the next result, which proves Theorem~\ref{thm:circ}.
\begin{corollary}\label{cor}
The word $g_3(b_4)$ avoids every circular formula $C_i$ with $i\ge3$.
The word $g_2(b_4)$ avoids every circular formula $C_i$ with $i\ge4$.
\end{corollary}

%
%

\section{Remaining formulas in the $3$-avoidance basis}

In this section, we prove the following result which completes the determination of the avoidability index of the formulas in the $3$-avoidance basis.
\begin{theorem}~\label{thm:two}
$\lambda(ABCBA.CBABC)=2$. $\lambda(ABCA.CABC.BCB)=3$.
\end{theorem}
Notice that $\lambda(ABCBA.CBABC)=2$ implies the well-known fact that $\lambda(ABABA)=2$.

For both formulas, we give a uniform morphism $m$ such that for every $\paren{\frac54^+}$-free word $w\in\Sigma_5^*$,
the word $m(w)$ avoids the formula. Since there exist exponentially many $\paren{\frac54^+}$-free words
over $\Sigma_5$~\cite{KR2011}, there exist exponentially many words avoiding the formula.
The proof that the formula is avoided follows the method in~\cite{Ochem2004}.


To avoid $ABCBA.CBABC$, we use this $15$-uniform morphism:
$$
\begin{array}{c}
m_{15}(\texttt{0})=\texttt{001111010010110}\\
m_{15}(\texttt{1})=\texttt{001110100101110}\\
m_{15}(\texttt{2})=\texttt{001101001011110}\\
m_{15}(\texttt{3})=\texttt{000111010001011}\\
m_{15}(\texttt{4})=\texttt{000110100001011}\\
\end{array}
$$
First, we show that the $m_{15}$-image of every $\paren{\frac54^+}$-free word $w$ is $\paren{\frac{97}{75}^+,61}$-free,
that is, $m_{15}(w)$ contains no repetition with period at least $61$ and exponent strictly greater than $\frac{97}{75}$.
By Lemma 2.1 in~\cite{Ochem2004}, it is sufficient to check this property for $\paren{\frac54^+}$-free word $w$ such that
$|w|<\frac{2\times\tfrac{97}{75}}{\tfrac{97}{75}-\tfrac54}<60$.
Consider a potential occurrence $h$ of $ABCBA.CBABC$ and write $a=|h(A)|$, $b=|h(B)|$, $c=|h(C)|$.
Suppose that $a+b\ge61$. The factor $h(BAB)$ is then a repetition with period $a+b\ge61$, so that its exponent satisfies
$\frac{a+2b}{a+b}\le\frac{97}{75}$. This gives $53b\le22a$. 
Similarly, $BCB$ implies $53b\le22c$, $ABCBA$ implies $53a\le22(2b+c)$, and $CBABC$ implies $53c\le22(a+2b)$.
Summing up these inequalities gives $53a+106b+53c\le44a+88b+44c$, which is a contradiction.
Thus, we have $a+b\le60$. By symmetry, we also have $b+c\le60$.
Using these inequalities, we check exhaustively that $h(w)$ contains no occurrence of $ABCBA.CBABC$.

To avoid $ABCA.CABC.BCB$ and its reverse $ABCA.BCAB.CBC$ simultaneously, we use this $6$-uniform morphism:
$$
\begin{array}{c}
m_6(\texttt{0})=\texttt{021210}\\
m_6(\texttt{1})=\texttt{012220}\\
m_6(\texttt{2})=\texttt{012111}\\
m_6(\texttt{3})=\texttt{002221}\\
m_6(\texttt{4})=\texttt{001112}\\
\end{array}
$$
We check that the $m_6$-image of every $\paren{\frac54^+}$-free word $w$ is $\paren{\frac{13}{10}^+,25}$-free.
By Lemma 2.1 in~\cite{Ochem2004}, it is sufficient to check this property for $\paren{\frac54^+}$-free word $w$ such that
$|w|<\frac{2\times\tfrac{13}{10}}{\tfrac{13}{10}-\tfrac54}=52$.

Let us consider the formula $ABCA.CABC.BCB$.
Suppose that $b+c\ge25$. Then $ABCA$ implies $7a\le3(b+c)$, $CABC$ implies $7c\le3(a+b)$, and $BCB$ implies $7b\le3c$.
Summing up these inequalities gives $7a+7b+7c\le3a+6b+6c$, which is a contradiction.
Thus $b+c\le24$. Suppose that $a\ge23$. Then $ABCA$ implies $a\le\frac37(b+c)\le\frac{72}{7}<23$, which is a contradiction.
Thus $a\le22$. For the formula $ABCA.BCAB.CBC$, the same argument holds except that the roles of $B$ and $C$ are switched,
so that we also obtain $b+c\le24$ and $a\le22$.
Then we check exhaustively that $h(w)$ contains no occurrence of $ABCA.CABC.BCB$ and no occurrence of $ABCA.BCAB.CBC$.

\section{Concluding remarks}\label{sec:con}
A major open question is whether there exist avoidable formulas with arbitrarily large avoidability index.
If such formulas exist, some of them necessarily belong to the $n$-avoidance basis for increasing values of $n$.
With the example of circular formulas, Clark noticed that belonging to the $n$-avoidance basis
and having many variables does not imply a large avoidability index.
Our results strengthen this remark and show that the $n$-avoidance basis contains a $2$-avoidable formula
on $t$ variables for every $3\le t\le n$.

Concerning conjugacy classes, we propose the following conjecture:
\begin{conjecture}
There exists an infinite word in $\Sigma_5^*$ that avoids every conjugacy class of length at least 2.
\end{conjecture}
Associated to the results in Lemma~\ref{L:1}, this would give the smallest alphabet that
allows to avoid every conjugacy class of length at least $i$, for every $i$.



\end{document}